\documentclass{elsarticle}  %use this for latex2e

\usepackage{amsmath}
\usepackage{graphicx}
\usepackage{amssymb}
\usepackage{amsfonts}
\usepackage{graphicx,color}
\usepackage{amsbsy}

\usepackage[british]{babel}
\usepackage[latin1]{inputenc}

\newtheorem{Def}{Definition}[section]

\newtheorem{Thm}{Theorem}[section]

\newtheorem{Ex}{Example}[section]

\newtheorem{proof}{Proof}

\newtheorem{Lem}{Lemma}[section]

\newtheorem{Rem}{Remark}[section]

\newcommand{\A}{\mathcal{A}_{p}}

\newcommand{\ZZ}{\Z_{p^r}}

\newcommand{\F}{{\mathbb F}}
\newcommand{\Z}{{\mathbb Z}}

\newcommand{\wt}{\mathrm{wt}}

\newcommand{\beq}{\begin{equation}}
\newcommand{\eeq}{\end{equation}}
\newcommand{\bmat}{\left[ \begin{array}}
\newcommand{\emat}{\end{array} \right]}

\newcommand{\im}{\mbox{ Im }}

\newcommand{\ie}{\emph{i.e.}}

\newcommand{\hs}{ }
\newcommand{\codc}{\mbox{$\mathcal{C}$}}
\newcommand{\bg}{$b$-degree}
\newcommand{\lzp}{left-zero prime}

\title{\LARGE \bf
Noncatastrophic convolutional codes over a finite ring
}

\author[focal]{D. Napp  }%\fnref{fn1}}
\ead{diegonapp@gmail.com}

\author[rvt]{R. Pinto }%\fnref{fn2}}
\ead{raquel@ua.pt}

\author[rvt2]{C. Rocha \corref{cor1}}%\fnref{fn2}}
\ead{sitagomes@gmail.com}

\cortext[cor1]{Corresponding author}

\address[focal]{Departament de Matemàtiques, Universitat d'Alacant, Spain. }
\address[rvt]{Department of Mathematics, University of Aveiro, Portugal.}
\address[rvt2]{Instituto Superior de Contabilidade e Administração de Coimbra, Instituto Politécnico de Coimbra, Portugal.}

\fntext[Pi]{This work is supported by The Center for Research and Development in Mathematics and Applications (CIDMA) through the Portuguese
Foundation for Science and Technology (FCT - Fundação para a Ciência e a Tecnologia), references UIDB/04106/2020 and UIDP/04106/2020. Diego Napp is partially supported by Ministerio de Ciencia e Innovación via the grant with ref. PID2019-108668GB-I00.
}

\begin{document}

\begin{abstract}
Noncatastrophic encoders are an important class of polynomial generator matrices of convolutional codes. When these polynomials have coefficients in a finite field, these encoders have been characterized are being polynomial left prime matrices. In this paper we study the notion of noncatastrophicity in the context of convolutional codes when the polynomial matrices have entries in a finite ring. In particular, we need to introduce two different notion of primeness in order to fully characterize noncatastrophic encoders over the finite ring $\ZZ$. The second part of the paper is devoted to investigate the notion of free and column distance in this context when the convolutional code is a free finitely generated $\ZZ$-module. We introduce the notion of $b$-degree and provide new bounds on the free distances and column distance. We show that this class of convolutional codes is optimal with respect to the column distance and to the free distance if and only if its projection on $\mathbb Z_p$ is.\end{abstract}

\maketitle

\section{Introduction}

The notion of primeness plays a central role in the polynomial matrix approach to several areas of pure and applied mathematics, such as systems and control theory or coding theory. In this paper we consider polynomial matrices over the finite ring $\mathbb Z_{p^r}$, where $p$ is a prime and $r$ an integer greater than $1$.
Our motivation for considering such a finite ring $\mathbb Z_{p^r}$ stems from applications in the area of error-correcting codes and in particular of convolutional codes over $\mathbb Z_{p^r}$.  Such a ring is useful for the so-called coded modulation scheme where the codewords in $\mathbb Z_{p^r}$ are mapped onto phase-shift-keying (PSK) modulation signals sets.  The mapping is such that distances between modulation points are preserved under additive operations in $\mathbb Z_{p^r}$, see \cite{SridharaF05} for more details. In \cite{massey89} Massey and Mittelholzer observed for the first time that convolutional codes over $\mathbb Z_{M}$, are the most appropriate class of codes for phase modulation. Note that even though we will focus on the ring $\mathbb Z_{p^r}$, by the Chinese Remainder Theorem, results on codes over $\mathbb Z_{p^r}$ can be extended to codes over $\mathbb Z_{M}$.\\

An important concept in the theory of convolutional codes is the noncatastrophicity.  When a catastrophic convolutional generator matrix is used for encoding, finitely many errors in the estimate of the transmitted codeword can lead to infinitely many errors in the estimate of the information sequence. This is of course a catastrophic situation that has to be avoided when designing the generator matrix. Noncatastrophic generator matrices have been characterized as left prime polynomial matrices and have been studied in several contexts depending on the definition considered in each case, see \cite{GianiraJulia2020,ku09,el13,rosenthal98}. In this work we define convolutional codes as finitely generated free $\mathbb Z_{p^r}[d]$-modules of $\mathbb Z_{p^r}[d]^n$, where $\mathbb Z_{p^r}[d]$ is the polynomial ring with coefficients in $\mathbb Z_{p^r}$ and study noncatastrophicity in this setting \cite{ku07,NaPiTo17,NaPiTo16,so07}. In the case of matrices with entries in $\mathbb Z_{p^r}[d]$ we need to distinguish two types of left primeness, namely, zero left prime and factor left prime, as happens in the case of polynomial matrices in several variables over a field \cite{LIN2001,NaRoSCL,Sha14,ZerPM}. We provide a characterization of zero left prime polynomial matrices from which it follows that when a convolutional code admit a left zero prime generator matrix, i.e., a noncatastrophic encoder, then the code can be described by means of a parity-check polynomial matrix. \\

The second part of the paper is devoted to investigating the Hamming distances of these codes as these will determine their error-correcting capabilities. In the context of convolutional codes the column distance is arguably the most important notion of distance \cite{jo99} and therefore we shall focus on the study of this particular distance. To this end we introduce a novel notion, called the $b$-degree, and derive bounds on the column distance in terms of the dimension, length and $b$-degree. In \cite{el13} it was proven that the free distance of convolutional codes over $\mathbb Z_{p^r}$ it is determined by its projection over $\mathbb Z_p$. The authors used this fact and
the Hensel lift of a cyclic code in \cite{el13} to build optimal convolutional codes with respect to the free distance. Here we show that  a convolutional code over $\mathbb Z_p$ is optimal with respect to the column distance if and only if its projection is. This will allow the construction of optimal convolutional codes over $\mathbb Z_{p^r}$ from well-known classes of convolutional codes over $\mathbb Z_{p}$.

The results of the paper are twofold: we first analyse the primeness of polynomial matrices with entries in $\mathbb Z_{p^r}[d]$ in Section \ref{sec:pol} and second we investigate the column distances of free convolutional codes over $\mathbb Z_{p^r}$ in Section \ref{sec:dist}. In each section we briefly provide some preliminaries:  in Section \ref{sec:pol} we recall known results of primeness of polynomial matrices over finite fields and in Section \ref{sec:dist} the definitions of convolutional codes, free distance and column distances are presented.

\section{Primeness of polynomial matrices over $\ZZ$}\label{sec:pol}

We denote by $\mathbb F[d]$ the ring of polynomials in the indeterminate $d$ and coefficients in a finite field $\mathbb F$ and by $\mathbb F(d)$ the field of rational functions defined in $\mathbb F$. Next we will present results that are well-known in the literature, see \cite{FoPi04,fo70} for more details.

\subsection{Primeness of polynomial matrices over a finite field $\F$}%%%%%%%%%%%%%%%%%%%%%%%%%%%%%%%%

\begin{Def} A polynomial matrix $U(d) \in \mathbb F[d]^{k \times k}$ is unimodular if it is invertible and $U(d)^{-1} \in \mathbb F[d]^{k \times k}$.
\end{Def}

\begin{Lem}
Let $U(d) \in \mathbb F[d]^{k \times k}$. Then $U(d)$ is unimodular if and only if det$\; U(d) \in \mathbb F \backslash\{0\}$.
\end{Lem}

\begin{Def}
A polynomial matrix $A(d) \in \mathbb F[d]^{k \times n}$ is left prime if in all factorizations
$$
A(d)=\Delta(d) \bar A(d), \;\mbox{with} \; \Delta(d) \in \mathbb F[d]^{k \times k}, \;\mbox{and} \; \bar A(d) \in \mathbb F[d]^{k \times n},
$$
the left factor $\Delta(d)$ is unimodular.
\end{Def}

Left prime matrices admit several characterizations as stated in the next theorem.

\begin{Thm} \label{lp}
Let $A(d) \in \mathbb F[d]^{k \times n}$. The following are equivalent:
\begin{enumerate}
\item $A(d)$ is left prime;
\item there exist unimodular matrices $U(d) \in \mathbb F[d]^{k \times k}$ and $V(d) \in \mathbb F[d]^{n \times n}$ such that
$$
U(d)A(d)V(d)=[I_k \;\; 0];
$$
\item there exists a unimodular matrix $V(d) \in \mathbb F[d]^{n \times n}$ such that $A(d)V(d)=[I_k \;\; 0]$;
\item there exists $B(d) \in \mathbb F[d]^{(n-k) \times n}$ such that $\left[\begin{matrix} A(d) \\ B(d) \end{matrix}\right]$ is unimodular;
\item $A(d)$ admits a polynomial right inverse;
\item for all $u(d) \in \mathbb F(d)^k$, $u(d)A(d) \in \mathbb F[d]^n$ implies that $u(d) \in \mathbb F[d]^k$;
\item $A(\alpha)$ has rank $k$ for all $\alpha \in \bar{\mathbb F}$, where $\bar{\mathbb F}$ denotes the algebraic closure of $\mathbb F$;
\item the ideal generated by all the $k$-th order minors of $A(d)$ is $\mathbb F[d]$.
\end{enumerate}
\end{Thm}

\subsection{Primeness of polynomial matrices over $\ZZ$}%%%%%%%%%%%%%%%%%%%%%%%%%%%%%%%%

In this section we study the notion of left prime for polynomial matrices over $\ZZ$. We denote by $\ZZ[d]$ the ring of polynomials in the indeterminate $d$, with coefficients in $\ZZ$ and by $\ZZ(d)$ the ring of rational functions defined, see \cite{jo98}, as the set
$$
\small
\left\{\frac{p(d)}{q(d)}: p(d),q(d) \in \ZZ[d] \mbox{ and the coefficient of the smallest power of $d$ in $q(d)$ is a unit}\right\}.
$$
%modulo the equivalence relation
This condition allows us to treat a rational function as an equivalence class in the relation
$$
\frac{p(D)}{q(D)} \sim \frac{p_1(D)}{q_1(D)} \mbox{ if and only if } p(D)q_1(D) = p_1(D) q(D).
$$

%Note that $\ZZ(D)$ is a subring of $\ZZ((D))$ and, obviously $\ZZ[D]$ is a subring of $\ZZ(D)$.

Any element $a \in \ZZ$ has a $p$-adic expansion \cite{ca00a}, \ie, it can be written uniquely as a linear combination of
$1,p,p^2,\dots$ $ \dots, p^{r-1}$, with coefficients in
$\A=\{0,1, \dots,p-1\} \subset \ZZ$,
$$
a=\alpha_0+\alpha_1p+\dots+\alpha_{r-1}p^{r-1}, \; \; \alpha_i\in\A, \; \;i=0,1,\dots,r-1.
$$
Note that all elements in $\A \backslash\{0\}$ are units. Given a matrix $A(d) \in \mathbb Z_{p^r}[d]^{k \times n}$, denote by $[A(d)]_p$ its (componentwise) projection into $\mathbb Z_p$.

\begin{Def}
A polynomial matrix $U(d) \in \ZZ[d]^{k \times k}$ is unimodular if it is invertible and \mbox{$U(d)^{-1} \in \ZZ[d]^{k \times k}$.}
\end{Def}

Next lemma characterizes the polynomial matrices over $\ZZ[d]$ which admit a right polynomial inverse.

\begin{Lem} A polynomial matrix $A(d) \in \ZZ[d]^{k \times n}$, with $n \geq k$, admits a polynomial right inverse if and only if $[A(d)]_p \in \Z_p[d]^{k \times n}$
also admits a polynomial right inverse over $\Z_p[d]$.
\end{Lem}

\begin{proof}
If $[A(d)]_p$ admits a polynomial right inverse over $\Z_p[d]$ then there exists $B(d) \in \Z_p[d]^{n \times k}$ such that
$$
[A(d)]_p B(d)=I_k
  \mod \; p.$$
Considering $B(d)$ as a matrix over $\ZZ[d]$, we have that
$$
A(d) B(d) = I_k-p X(d),
$$
for some $X(d) \in \ZZ[d]^{k \times k}$ and therefore
$$
B(d)(I_k + pX(d) + p^2 X^2(d) + \cdots + p^{r-1}X^{r-1}(d))
$$
is a right inverse of $A(d)$. The converse is obvious.
\end{proof}

The following theorem is immediate.

\begin{Thm}\label{unimodular}%\cite{ku07}
Let $U(d) \in \ZZ[d]^{k \times k}$. The following are equivalent:
\begin{enumerate}
\item $U(d)$ is unimodular;
\item $det \, U(d)$ is a unit;
\item $[U(d)]_p$ is unimodular in $\Z_p[d]^{k \times k}$.
\end{enumerate}
\end{Thm}

Left primeness is a property of polynomial matrices which plays a fundamental role when we consider convolutional codes over a finite field $\mathbb F$. As mentioned before, a polynomial matrix \mbox{$A(d) \in \mathbb F[d]^{k \times n}$} is left prime if in all factorizations
$$
A(d)=\Delta(d) \tilde A(d), \;\mbox{with}\; \Delta(d) \in \mathbb F[d]^{k \times k}, \;\mbox{and}\; \tilde A(d) \in \mathbb F[d]^{k \times n},
$$
the left factor $\Delta(d)$ is unimodular, or equivalently if the ideal generated by all the $k$-th order minors of $A(d)$ is $\mathbb F[d]$ (see Theorem \ref{lp}). However, this equivalence does not hold over $\ZZ$. There are polynomial matrices over $\ZZ[d]$ that satisfy the former condition but do not satisfy the later, as it is illustrated in the following example.

\begin{Ex}\label{ex1}
The matrix
$$
A(d) = \left[\begin{matrix} 1 + 3d & 1 + d \end{matrix} \right] \in \mathbb Z_4[d]^2
$$
does not have a nonunimodular left factor, but the ideal generated by its full size minors is
$$\{(1+d)p(d) \, : \; p(d) \in \ZZ[d]\}.$$
\end{Ex}

Therefore, we need to introduce two different notions of primeness when dealing with polynomial matrices over $\ZZ$.

\begin{Def}
A polynomial matrix $A(d) \in \ZZ[d]^{k \times n}$ is left factor-prime ($\ell FP$) if in all factorizations
$$
A(d)=\Delta(d) \bar A(d)\mbox{ with } \Delta(d) \in \ZZ[d]^{k \times k}\mbox{ and }\bar A(d) \in \ZZ[d]^{k \times n},
$$
the left factor $\Delta(d)$ is unimodular.
\end{Def}

\begin{Def}
A polynomial matrix $A(d) \in \ZZ[d]^{k \times n}$ is left zero-prime ($\ell ZP$) if the ideal generated by all the $k$-th order minors of $A(d)$ is $\ZZ[d]$.
\end{Def}

Right factor-prime (rFP) and right zero-prime (rZP) matrices are defined in the same way, upon taking transposes.

\begin{Rem}
  Note that the fact that the conditions of Theorem \ref{lp} are not anymore equivalent when considering rings instead of fields also occurs when considering the polynomial ring $\F[d_1,\dots , d_n]$ in several variables instead of $\mathbb F[d]$, see \cite{LIN2001,chinos2020,NaRoSCL,Sha14,ZerPM} for more details.
\end{Rem}

As shown in Example \ref{ex1} factor-primeness does not imply zero-primeness, however the converse is true as stated in the following lemma.

\begin{Lem}
Let $A(d) \in \ZZ[d]^{k \times n}$. If $A(d)$ is left zero-prime then it is also left factor-prime.
\end{Lem}

\begin{proof}
Let us assume that $A(d)$ is not left factor prime. Then $A(d)=X(d) \tilde A(d) $ for some \mbox{$\tilde A(d) \in \ZZ[d]^{k \times n}$} and $X(d) \in \ZZ[d]^{k \times k}$ nonunimodular. Then by Theorem \ref{unimodular} $x(d) = det \, X(d)$ is not a unit and the ideal generated by all the $k$-th order minors of $A(d)$ is contained in \mbox{$\{x(d)p(d) \, : \, p(d) \in \ZZ[d]\}$}. Consequently, $A(d)$ is not left zero-prime.
\end{proof}

It is easy to see that an ideal $\mathcal{I}$ of $\ZZ[d]$ is equal to $\ZZ[d]$ if and only if
$[\mathcal{I}]_p=\{[u]_p:\;u\in\mathcal{I}\}$ is equal to $\Z_p[d]$ and, therefore, the next lemma follows immediately.

\begin{Lem}\label{lzplp}
$A(d) \in \ZZ[d]^{k \times n}$ is left zero-prime over $\ZZ[d]$ if and only if $[A(d)]_p \in \mathbb Z_p[d]^{k \times n}$ is left prime over $\mathbb Z_p[d]$.
\end{Lem}

Now, we are in position to prove the following characterizations of left zero-prime matrices with entries in $\ZZ$ which can be considered as an extension of Theorem \ref{lp} to the finite ring case.

\begin{Thm}\label{lzp}
Let $A(d) \in \ZZ[d]^{k \times n}$. The following are equivalent:
\begin{enumerate}
\item $A(d)$ is left zero-prime;
\item there exist unimodular matrices $U(d) \in \ZZ[d]^{k \times k}$ and $V(d) \in \ZZ[d]^{n \times n}$ such that \mbox{$U(d)A(d)V(d)=[I_k \;\; 0]$};
\item there exists a unimodular matrix $V(d) \in \ZZ[d]^{n \times n}$ such that $A(d)V(d)=[I_k \;\; 0]$;
\item there exists $B(d) \in \ZZ[d]^{(n-k) \times n}$ such that $\left[\begin{matrix} A(d) \\ B(d) \end{matrix}\right]$ is unimodular;
\item $A(d)$ admits a polynomial right inverse;
\item for all $u(d) \in \ZZ(d)^k$, $u(d)A(d) \in \ZZ[d]^n$ implies that $u(d) \in \ZZ[d]^k$;
\item $\bar A(\alpha)$ has rank $k$,$\mod p$, for all $\alpha \in \bar{\mathbb Z}_p$, where $\bar {\mathbb Z}_p$ denotes the algebraic closure of $\mathbb Z_p$ and $\bar A(d)=[A(d)]_p$.
\end{enumerate}
\end{Thm}

\begin{proof}
  From Theorems \ref{lp} and \ref{unimodular} and Lemma \ref{lzplp} we immediately conclude that $2) \, \Rightarrow 1)$, $3) \, \Rightarrow 1)$, $4) \, \Rightarrow 1)$, $5) \, \Rightarrow 1)$, $7) \, \Rightarrow 1)$ and $1) \, \Rightarrow 7)$. Next we prove the implications $1) \, \Rightarrow 2)$, $2) \, \Rightarrow 3)$, $3) \Rightarrow 4)$, $4) \Rightarrow 5)$, $5) \Rightarrow 6)$ and $6) \Rightarrow 1)$.

\vspace{.2cm}

$1) \, \Rightarrow 2)$: Since $A(d)$ is $\ell ZP$, $[A(d)]_p$ is left prime over $\mathbb Z_p[d]$ and therefore there exist unimodular matrices $U(d) \in \mathbb Z_p[d]^{k \times k}$ and $V(d) \in \mathbb Z_p[d]^{n \times n}$ such that
\[
U(d)[A(d)]_pV(d) = [I_k \;\; 0] \mod \, p.
\]
Considering $U(d)$ and $V(d)$ as matrices over $\mathbb Z_{p^r}[d]$ we have that
\[
U(d)A(d)V(d)=[X_1(d) \;\; X_2(d)],
\]
with $X_1(d) \in \ZZ[d]^{k \times k}$ and $X_2(d) \in \ZZ[d]^{k \times (n-k)}$. Note that, $X_1(d)$ is unimodular because $[X_1(d)]_p=I_k$ and that $U_1(d)=X_1(d)^{-1}U(d)$ and $$
V_1(d)=V(d)\left[ \begin{array}{cc}
I_k & -X_1(d)^{-1}X_2(d) \\
  0 & I_{n-k}\end{array}\right]
$$ are polynomial matrices. It is easy too see that $U_1(d)$ and $V(d)$ are unimodular matrices, and that
\[
U_1(d)A(d)V_1(d)=[I_k \;\; 0].
\]

\vspace{.2cm}

$2) \Rightarrow 3)$: Let $U(d) \in \ZZ[d]^{k \times k}$ and $V(d) \in \ZZ[d]^{n \times n}$ be unimodular matrices such that \mbox{$U(d)A(d)V(d)=[I_k \;\; 0]$}. Then $A(d)V(d)=[U(d)^{-1}\;\; 0]$ and, therefore
$$
V_1(d)=V(d)\left[ \begin{array}{cc}
U(d) & 0 \\
  0 & I_{n-k}\end{array}\right]
$$
is a unimodular matrix such that $A(d)V_1(d)=[I_k \;\; 0]$.

\vspace{.2cm}

$3) \Rightarrow 4)$: From the assumption $A(d)=[I_k \; 0] \tilde V(d)$ for some unimodular matrix $\tilde V(d) \in \ZZ[d]^{n \times n}$, i.e., $A(d)$ is the submatrix of $\tilde V(d)$ constituted by its first $k$ rows.

\vspace{.2cm}

$4) \Rightarrow 5)$: Let $[X(d) \;\; Y(d)]$ with $X(d) \in \ZZ[d]^{n \times k}$ and $Y(d) \in \ZZ[d]^{n \times (n-k)}$ be such that
\[
\left[
\begin{array}{c}
A(d) \\ B(d)
\end{array}
\right]
\left[
\begin{array}{cc}
X(d) & Y(d)
\end{array}
\right]=I_n.
\]
Then $A(d)X(d)=I_k$.

\vspace{.2cm}

$5) \Rightarrow 6)$: Let $u(d) \in \ZZ(d)^k$ be such that $u(d)A(d)=w(d) \in \ZZ[d]^n$ and let $X(d) \in \ZZ[d]^{n \times k}$ be a right inverse of $A(d)$. Then $u(d)=w(d)X(d)$, which is a polynomial vector.

\vspace{.2cm}

$6) \Rightarrow 1)$: Let us assume that $A(d)$ is not $\ell ZP$. Then $[A(d)]_p$ is not left prime over $\mathbb Z_p[d]$, and therefore there exists a nonpolynomial vector $u(d) \in \mathbb Z_p(d)^k$ such that $u(d)[A(d)]_p \in \mathbb Z_p[d]^n \mod p$. Considering $u(d)$ as a vector over $\mathbb Z_{p^r}[d]$, it follows that $p^{r-1}u(d) \in \ZZ[d]^k$ is also nonpolynomial and $p^{r-1}u(d)A(d) \in \ZZ[d]^n$. \end{proof}

\section{Distance properties of free convolutional codes over $\ZZ$}\label{sec:dist}

In this section we first recall the basic definitions of convolutional codes over $\ZZ$. We consider convolutional codes as free $\ZZ[d]$-submodules of $\ZZ[d]^n$, for some $n \in \mathbb N$, see \cite{ku07,NaPiTo17,NaPiTo16,so07}. We require the encoding map to be injective and therefore focus on free submodules of $\ZZ[d]^n$. We note that different definitions have been considered in the literature, see for instance \cite{noemi17,ku08,ku09,el13}. The non-free case lies beyond the scope of this work but it can also be treated using the theory of $p$-basis and $p$-generating sequences, see for instance \cite{ku09,KuPi17,NaPiTo17,el13}.

\subsection{Convolutional codes}

\begin{Def}
A convolutional code $\cal C$ of rate $k/n$ is a free submodule of $\ZZ[d]^n$ of rank $k$. A matrix $G(d) \in \ZZ[d]^{k \times n}$ whose rows form a basis of $\cal C$ is called an encoder of $\cal C$.
\end{Def}

Thus, an encoder of $\cal C$ is a full row rank matrix $G(d) \in \ZZ[d]^{k \times n}$ such that
\begin{eqnarray*}
{\cal C} & = & \im_{\ZZ[d]} G(d)  =  \{u(d)G(d) \, : \, u(d) \in \ZZ[d]^k\}.
\end{eqnarray*}

%\begin{Rem}
%In [] convolutional codes are defined as general submodules of $\ZZ[d]^n$, but we opted to restrict to free submodules since we will focus on these codes in this paper.
%\end{Rem}

Equivalent encoders are full row rank matrices that are encoders of the same code. Then two equivalent encoders $G_1(d), G_2(d) \in \ZZ[d]^{k \times n}$ are such that $G_2(d)=U(d)G_1(d)$ for some unimodular matrix $U(d) \in \ZZ[d]^{k \times k}$.
%
%\begin{Rem}
%A convolutional code $\tilde {\cal C}$ of rate $k/n$ over a finite field is defined in the similar way and its encoders also differ on the left by a unimodular matrix.
%This implies that two encoders of $\tilde {\cal C}$ have the same full size minors, up to multiplication by a nonzero constant. The maximum degree of the full size minors of an encoder of $\tilde {\cal C}$ is called the degree of $\tilde {\cal C}$ and it is usually represented by $\delta$.
%\end{Rem}
%
Thus, it follows that if a convolutional code admits a left zero-prime encoder then all its encoders are also left zero-prime. We call such codes \emph{noncatastrophic} codes and they are the ones that admit a kernel representation as stated in the following theorem. %The proof follows the same reasoning as the counterpart result for convolutional codes over a finite field (Theorem 3.2.4 in \cite{York87}) and hence we omit it.

\begin{Thm}
Let $\cal C$ be a convolutional code of rate $k/n$. Then, there exists a full column rank polynomial matrix $H(d) \in \ZZ[d]^{n \times (n-k)}$ such that
\begin{eqnarray*}
{\cal C} & = & \ker_{\ZZ[d]} \, H(d)  = \{w(d) \in \ZZ[d]^n \, : \, w(d)H(d) = 0\}
\end{eqnarray*}
if and only if $\cal C$ is noncatastrophic.
\end{Thm}

\begin{proof}
Let us assume first that $\cal C$ is noncatastrophic and let $G(d) \in \ZZ[d]^{k \times n}$ be an encoder of $\cal C$. Then, $G(d)$ is $\ell ZP$ and therefore, by Theorem \ref{lzp}, there exist polynomial matrices $B(d) \in \mathbb F[d]^{(n-k)\times n}$, $X(d) \in \mathbb F[d]^{n \times k}$ and $H(d) \in \mathbb F[d]^{n \times (n-k)}$ such that
\[
\left[\begin{matrix} G(d) \\ B(d) \end{matrix}\right] \left[\begin{matrix} X(d) & H(d) \end{matrix}\right]=\left[\begin{matrix} X(d) & H(d) \end{matrix}\right] \left[\begin{matrix} G(d) \\ B(d) \end{matrix}\right]=I_n
\]
This means that $H(d)$ is a full column rank matrix such that $G(d)H(d)=0$, and therefore ${\cal C} \subset \ker_{\ZZ[d]} \, H(d)$. On the other hand, if $w(d) \in \ker_{\ZZ[d]} \, H(d)$ we have that
\begin{eqnarray*}
w(d) & = & w(d)\left[\begin{matrix} X(d) & H(d) \end{matrix}\right] \left[\begin{matrix} G(d) \\ B(d) \end{matrix}\right]\\
& = & \left[\begin{matrix}  w(d) X(d) & 0 \end{matrix}\right] \left[\begin{matrix} G(d) \\ B(d) \end{matrix}\right] \\
& = & u(d)G(d),
\end{eqnarray*}
where $u(d)= w(d) X(d) \in \ZZ[d]^k$, i.e. $w(d) \in \cal C$.

For the converse let us assume that $\cal C$ is not a noncatastrophic convolutional code  and that ${\cal C} = \ker_{\ZZ[d]} \, H(d)$ for some full column rank matrix $H(d) \in \ZZ[d]^{n \times (n-k)}$. Let $G(d)$ be an encoder of $\cal C$. Then, since $G(d)$ is not left zero-prime,
\[
[G(d)]_p = X(d) \tilde G(d) \mod p
\]
for some invertible but nonunimodular matrix $X(d) \in \mathbb Z_p[d]^{k \times k}$ and $\tilde G(d) \in \mathbb Z_p[d]^{k \times n}$. Considering $X(d)$ and $\tilde G(d)$ as matrices over $\mathbb Z_{p^r}[d]$, we have that
$$p^{r-1}G(d)=p^{r-1}X(d) \tilde G(d)
$$
and therefore
$$
p^{r-1}X(d)^{-1}G(d)=p^{r-1}\tilde G(d).
$$
 Since $X(d)$ is not unimodular, there exists an $i \in \{1, \dots,k\}$ such that the $i$-th row of $p^{r-1}X(d)^{-1}$ is not polynomial. Let us represent such row by $\ell_i(d)$, i.e., $\ell_i(d)=e_ip^{r-1}X(d)$, where $e_i$ is the $i$-th vector of the canonical basis. Then $\ell_i(d)G(d)$ does not belong to $\cal C$ because $\ell_i(d)$ is not polynomial, but since $\ell_i(d)G(d)=e_ip^{r-1}\tilde G(d)$ is a polynomial vector, it follows that   $\ell_i(d)G(d)$ belongs to $\ker_{\ZZ[d]} \, H(d)$, which is a contradiction.
\end{proof}

If $\cal C$ is a noncatastrophic convolutional code, then a full column rank polynomial matrix $H(d)$ such that ${\cal C} = \ker_{\ZZ[d]} H(d)$ is called a \emph{parity-check matrix} of $\cal C$.

\vspace{.2cm}

We conclude this section by giving a result on the relation between the order of an information sequence $u(d)$ and the corresponding codeword $w(d)=u(d)G(d)$ where $G(d)$ is an encoder. This relation will be useful later on the paper.

\vspace{.2cm}

Let $a \in \ZZ$. We define the order of $a$ to be $\ell$, and we write $ord(a)=\ell$, if the set $a \ZZ$ has $p^{\ell}$ elements. Then, $ord(a)=\ell$ if and only if $p^{\ell-1}a$ is a nonzero element of $p^{r-1} \ZZ$ and $p^{\ell}a=0$. In the same way we define the order of a polynomial vector $w(d) \in \ZZ[d]^m$ to be $\ell$, and we write $ord(w)=\ell$, if $p^{\ell-1} w(d) \neq 0$ and $p^{\ell} w(d)=0$. This means that $p^{\ell-1}w(d)$ is a nonzero element of $p^{r-1}\ZZ[d]^m$. The following lemma relates the orders of an information sequence and the corresponding codeword. We omit the simple proof.

\begin{Lem} \label{ord}
Let $\cal C$ be a convolutional code of rate $k/n$, $G(d)$ an encoder of $\cal C$ and $w(d)=u(d)G(d)$, with $u(d) \in \ZZ[d]^k$, a codeword of $\cal C$. Then
$$
ord(w)=ord(u).
$$
\end{Lem}
%
%\pf
%Let $\ell$ be the order of $u(d)$. Then $p^{\ell}w(d) = p^{\ell} u(d)G(d)=0$, since $p^{\ell}u(d)=0$, and $p^{\ell-1}w(d) = p^{\ell-1} u(d)G(d) \neq 0$ because $ p^{\ell-1} u(d) \neq 0$ and $G(d)$ is full column rank. We conclude that $w(d)$ has also order $\ell$.
%\pfend\eind

%\begin{Lem}
%Let $\cal C$ be a basic convolutional code of rate $k/n$ and $G(d)$ an encoder of $\cal C$ and $v(d)=u(d)G(d)$, with $u(d) \in \ZZ[d]^k$, a codeword of $\cal C$. Then
%$$
%ord(v)=ord(v).
%$$
%\end{Lem}

\subsection{Distance properties}

Next we study the free distance and column distances of a convolutional code over $\ZZ$. Such distances were also investigated in \cite{NaPiTo17,NaPiTo16,el13} for not necessarily free convolutional codes using the notion of $p$-basis, see \cite{NaPiTo16,el13}. For the free case addressed in this work we introduce the notion of $b$-degree of a code and derive new bounds on the free and column distance in terms of the length, dimension and $b$-degree of the code. First, we formally present the definitions of free distance and column distance. \\

The free distance of a convolutional code is defined as
$$
d_{\rm free}({\cal C})= {\rm min}\{\wt(w(d)) \, : \, w(d) \in {\cal C}, \, w(d) \neq 0\}
$$
where for $w(d)= \sum_{i \in \mathbb N_0} w_i d^i$, $\wt(w(d))= \sum_{i \in \mathbb N_0} \wt(w_i)$, with $\wt(w_i)$ to be the number of nonzero entries of $w_i$. Let $[{\cal C}]_p=\lbrace [w(d)]_p:\;w(d)\in \cal C  \rbrace$ and define $$d_{\rm free}([{\cal C}]_p)=min\lbrace wt(v(d)):\; v(d) \in[{\cal C}]_p,\, v(d)\neq 0 \rbrace.$$
In \cite[Theorem 5.3]{el13} it was shown that
\begin{equation} \label{dfree}
d_{\rm free}({\cal C}) \geq d_{\rm free}([{\cal C}]_p).
\end{equation}
This can be alternatively shown as follows. Note that $[{\cal C}]_p \simeq p^{r-1} {\cal C}$ and that $d_{\rm free}([{\cal C}]_p)= d_{\rm free}(p^{r-1} {\cal C})$ with
$d_{\rm free}(p^{r-1} {\cal C})= {\rm min}\{\wt(p^{r-1}w(d)) \, : \, w(d) \in {\cal C}, \, [w(d)]_p \neq 0\} $. Let $w(d)$ be a nonzero codeword of ${\cal C}$ of order $\ell$. Lemma \ref{ord} implies that $p^{\ell-1} w(d)$ is a nonzero vector of $p^{r-1}{\cal C}$, and since
$$
\wt(w(d)) \geq \wt(p^{\ell-1} w(d))
$$
the inequality (\ref{dfree}) follows. Next theorem shows that inequality (\ref{dfree}) is in fact an equality.

\begin{Thm}\label{thmequaldfree}
Let ${\cal C}$ be a convolutional code. Then
$$
d_{\rm free}({\cal C}) = d_{\rm free}([{\cal C}]_p).
$$
\end{Thm}

\begin{proof}
We only need to prove that
$$
d_{\rm free}({\cal C}) \leq d_{\rm free}([{\cal C}]_p).
$$
For that let $w(d)$ be a nonzero codeword of $[{\cal C}]_p$. Then there exists $\tilde w(d) \in {\cal C}$ such that $[{\tilde w}(d)]_p=w(d)$. Thus, $p^{r-1}\tilde w(d) \in {\cal C}$ with $\wt(p^{r-1} \tilde w(d))=\wt(w(d))$, which implies that $d_{\rm free}({\cal C}) \leq d_{\rm free}([{\cal C}]_p)$.
\end{proof}

The maximum value that the free distance of a convolutional code over a finite field of rate $k/n$ can attain depends also of the degree of the code which is defined as the
maximum of the degrees of the determinants of the submatrices of one and hence any generator matrix of ${\cal C}$. If ${\cal C}$ is a convolutional code over a finite field of rate $k/n$ and degree $\delta$, then
\begin{eqnarray}
d_{\rm free}({{ \cal C}}) \leq (n-k) \left(\lfloor \frac{\delta}{k}\rfloor+1\right) + \delta + 1. \label{dfreemaj}
\end{eqnarray}
This upper bound is called the Generalized Singleton bound and was found first in \cite{ro99} in the field case and then extended in \cite{NaPiTo16,el13} for the ring case using the notion of $p$-degree that is not used in this work.

For the case of free modules of $\ZZ^n$ that is considered here we will obtain a new expression for the bound on the free distance of a free convolutional code ${\cal C}$. For that we need to introduce the novel concept of $b$-degree of ${\cal C}$.

\begin{Def}\label{bdegree}
Let ${\cal C}$ be a convolutional code over $\ZZ[d]$. The $b$-degree of ${\cal C}$ is equal to the degree of $[{{\cal C}}]_p$.
\end{Def}

The $b$-degree of a convolutional code ${\cal C}$ can be easily obtained by calculating the maximum degree of the full size minors of $[G(d)]_p$ mod $p$, for any encoder $G(d)$ of ${\cal C}$.

The next result follows immediately from Theorem \ref{thmequaldfree}, Definition \ref{bdegree} and from the expression (\ref{dfreemaj}).
%
%The free distance of a convolutional code ${\cal C}$ of rate $k/n$ and $b$-degree $\delta$ is upper bounded by the generalized Singleton bound on the free distance of  the convolutional of rate $k/n$ and degree $\delta$ defined over a finite field \cite{ro99a1} as stated in the following theorem.

\begin{Thm}
Let ${\cal C}$ be a convolutional code of rate $k/n$ and $b$-degree $\delta$. Then
$$
d_{\rm free}({\cal C}) \leq (n-k) \left(\lfloor \frac{\delta}{k}\rfloor+1\right) + \delta + 1.
$$
\end{Thm}

%\begin{Rem}
%This bound is more accurate than the generalized Singleton bound obtained in [....] where the authors considered convolutional codes as general submodules of $\ZZ[d]^n$.
%\end{Rem}

A convolutional code of rate $k/n$ and $b$-degree $\delta$ with free distance $(n-k) \left(\lfloor \frac{\delta}{k}\rfloor+1\right) + \delta + 1$ is said to be a Maximum Distance Separable (MDS) code. It follows immediately from Theorem \ref{thmequaldfree} that a convolutional code $\cal C$ is MDS if and only if $[{\cal C}]_p$ is also MDS over $\mathbb Z_p[d]$.

%\section{Column distances}

\vspace{.2cm}

Another type of distances of a convolutional code which can be very useful in sequential decoding and have showed to have a potential use in streaming applications are the column distances \cite{Lieb2018}.

\vspace{.2cm}

Let $G(d)$ be an encoder of $\cal C$ and let us write
$
G(d)=\sum_{i=0}^{\nu} G_i d^i, \;\;\; G_i \in \ZZ^{k \times n}.
$
The codeword $w(d)=\sum_{i \in \mathbb N_0} w_id^i$, $w_i \in \ZZ^n$, corresponding to $u(d)=\sum_{i \in \mathbb N_0} u_i d^i$, $u_i \in \ZZ^k$, is such that
\[
[w_0 \; w_1 \; \cdots \; w_j]= [u_0 \; u_1 \; \cdots \; u_j] G^c_j
\]
where
\[
G^c_j=\left[
\begin{array}{cccc}
G_0 & G_1 & \cdots & G_j \\
& G_0 & \dots & G_{j-1} \\
& & \ddots & \\
& & & G_0
\end{array}
\right]
\]
is called the truncated sliding matrix corresponding to $G(d)$ (we consider $G_j=0$, if $j > \nu$), \cite{gl03,NaPiTo17}.

\begin{Def}\cite{NaPiTo17,TosteTese}
Given an encoder $G(d)$ of a convolutional code $\cal C$, we define the $j$-th column distance of $G(d)$
as
\[
d^c_j(G)=\min\{\wt([u_0 \; u_1 \; \cdots \; u_j] G^c_j): u_i \in \ZZ^k, \; u_0 \neq 0\}.
\]
\end{Def}

Parity-check matrices are very useful in the analysis of such distances. For this reason we restrict the study of such distances to noncatastrophic convolutional codes. Note that if $G(d)$ is an encoder of a noncatastrophic convolutional code $\cal C$, then $G(d)$ has a right polynomial inverse and therefore $G(0)$ is full row rank and this means that the $j$-th column distance of $\cal C$ is an invariant of the code and can be obtained as
\[
d^c_j({\cal C})=\min\{\wt([w_0 \; w_1 \; \cdots \; w_j]): \; [w_0 \; w_1 \; \cdots \; w_j] \in \im \, G^c_j, \; w_0 \neq 0\}.
\]
If $H(d)=\sum_{i=0}^{\ell} H_i d^i, \; H_i \in \ZZ^{(n-k) \times n}$ is a parity-check of ${\cal C}$ and $H(\ell)\neq 0$, $\ell\in\mathbb{N}$, then
\[
d^c_j({\cal C})=\min\{\wt([w_0 \; w_1 \; \cdots \; w_j]): \; [w_0 \; w_1 \; \cdots \; w_j] (H^c_j)=0, \; w_0 \neq 0\},
\]
where
\[
H^c_j=\left[
\begin{array}{cccc}
H_0 & H_1& \cdots & H_j\\
& H_0 & \cdots & H_{j-1} \\
 & & \ddots & \vdots \\
& &  & H_0
\end{array}
\right],
\]
with $H_j=0$ for $j > \ell$. The following theorem is immediate and its proof is analogous to the field case.

\begin{Thm}\label{indepTruncDes} Let ${\cal C}$ be a noncatastrophic convolutional code of rate $k/n$. Then, for any $j,d \in \mathbb N$, $d^c_j({\cal C}) = d$ if and only if the following conditions are satisfied:
\begin{enumerate}
\item there exist $d$ rows of $H^c_j$ linearly dependent over $\ZZ[d]$ such that one of these rows belongs to the first $n$ rows of $H^c_j$;
\item all $d-1$ rows, in which one of the rows belongs to the first $n$ rows of $H^c_j$, are linearly independent over $\ZZ[d]$.
\end{enumerate}
\end{Thm}
\begin{Rem}\label{rem}
Since the lines of a matrix $A$ are linearly independent if and only if $[A] _p$ is a full row rank, the conditions $1$ and $2$ of the Theorem \ref{indepTruncDes} can be expressed, respectively, and in an equivalent way, as follows:
\begin{enumerate}
\item there exist $d$ rows of $[H^c_j]_p$ linearly dependent over $\Z_p[d]$ such that one of these rows belongs to the first $n$ rows of $[H^c_j]_p$;
\item all $d-1$ rows, in which one of the rows belongs to the first $n$ rows of $[H^c_j]_p$, are linearly independent over $\Z_p[d]$.
\end{enumerate}
\end{Rem}
Next theorem provides upper bounds on the column distances of a noncatastrophic convolutional code. These upper bounds were found in \cite{NaPiTo17} for the more general case in which the convolutional codes are not necessarily noncatastrophic free convolutional codes. Although the result is not new, we opted to present the proof, because it is much simpler than the one in \cite{NaPiTo17}.

\begin{Thm}\label{rosenthal corpo}%\cite{NaPiTo17}
Let $\cal C$ be a noncatastrophic convolutional code of rate $k/n$. Then
\[
d^c_j({\cal C}) \leq (n-k)(j+1)+1,
\]
for all $j \in \mathbb N_0$.
\end{Thm}

\begin{proof}
Since \codc\hs is a noncatastrophic convolutional code over $\ZZ$, then$[{\cal C}]_p$ is a noncatastrophic convolutional code code over $\mathbb Z_p$ of rate $k/n$, and therefore
\[
d^c_j([{\cal C}]_p) \leq (n-k)(j+1)+1,
\]
for all $j \in \mathbb N_0$ (see \cite{gl03}). Let $w(d)=\sum_{i \in \mathbb N_0} w_i d^i \in [{\cal C}]_p$ with $w_0 \neq 0$, then
$$
\wt([w_0 \; w_1 \; \cdots \; w_j]) \leq (n-k)(j+1)+1.
$$
Let us consider $[w_0 \; w_1 \; \cdots \; w_j]$ as a vector of $\Z_{p^r}^{n(j+1)}$. Then, $\tilde w(d)=p^{r-1}w(d) \in p^{r-1} {\cal C}$ is such that $ \tilde w(0) \neq 0$ and $\wt([\tilde w_0 \; \tilde w_1 \; \cdots \; \tilde w_j])=\wt([w_0 \; w_1 \; \cdots \; w_j])$. Then $d^c_j({\cal C}) \leq (n-k)(j+1)+1$.
\end{proof}

The next result readily follows from \cite[Theorem 2.4]{gl03} and Remark \ref{rem}.

%\begin{Cor}
%Let $\cal C$ be a noncatasrophic convolutional code of rate $k/n$. If, for \mbox{$j\in\mathbb{N}_0$}, \mbox{$d^c_j({\cal C}) = (n-k)(j+1)+1$}, then, for $i\leq j$,
%\begin{equation}\label{imouij}
%d^c_i({\cal C}) = (n-k)(i+1)+1.
%\end{equation}
%\end{Cor}
%Theorema 333 linha 2515

\begin{Thm}
Let $G(d)$ be an encoder of a noncatastrophic convolutional code over $\ZZ[d]$, \codc, of rate $k/n$, $k\leq n$, and $H(d)$ be a parity-check matrix of \codc.
The following are equivalent:
\begin{enumerate}
\item $d^c_j({\cal C}) = (n-k)(j+1)+1$.
\item every $(j+1)k\times (j+1)k$ full-size minor of $[G_j^c]_p$ formed from the columns with indices \mbox{$1\leq t_1 < \ldots <t_{(j+1)k}$}, where $t_{sk+1}> sn$, $s=1,\, \ldots\, , j$,is nonzero.
%todo o  menor de maior ordem de $[G_j^c]_p$ de ordem, $(j+1)k\times (j+1)k$, que se obt\'{e}m de $[G_j^c]_p$ considerando as colunas de \'{\i}ndices $1\leq t_1 < \ldots <t_{(j+1)k}$, onde $t_{sk+1}> sn$, $s=1,\, \ldots\, , j$, \'{e} diferente de zero;
\item every $(j+1)(n-k)\times (j+1)(n-k)$ full-size minor of $[H_j^c]_p$ formed from the columns with indices $1\leq r_1 < \ldots <r_{(j+1)(n-k)}$, where $r_{s(n-k)}\leq sn$, $s=1, \ldots\, , j$, is nonzero.

%todo o  menor de maior ordem de $[H_j^c]_p$ de ordem, $(j+1)(n-k)\times (j+1)(n-k)$, que se obt\'{e}m de $[H_j^c]_p$ considerando as colunas de \'{\i}ndices $1\leq r_1 < \ldots <r_{(j+1)(n-k)}$, onde $r_{s(n-k)}\leq sn$, $s=1, \ldots\, , j$, \'{e} diferente de zero.
%\bigskip
\end{enumerate}
\end{Thm}
The column distances of a noncatastrophic \codc\hs do not grow indefinitely,
since they are naturally upper bounded by the free distance of \codc. If \codc\hs is a noncatastrophic convolutional code over a finite field of rate $k/n$ and degree $\delta$, then \codc\hs can have maximum column distances up to the $L$-\textit{th} column distance, where $L=\left\lfloor \dfrac{\delta}{k}\right\rfloor + \left\lfloor \dfrac{\delta}{n-k}\right\rfloor$. $L$ is the largest integer for which
$$
(n-k)(L+1)+1\leq (n-k)\left( \left\lfloor \dfrac{\delta}{k}\right\rfloor +1\right)+\delta+1.
$$

\begin{Def}\cite{gl03}
Let $\codc$ be a noncatastrophic code over a finite field $\mathbb{F}$, of rate $k/n$ and degree $\delta$.
Let \mbox{$L=\left\lfloor \dfrac{\delta}{k}\right\rfloor + \left\lfloor \dfrac{\delta}{n-k}\right\rfloor$}.
\codc\hs is a Maximum Distance Profile (MDP) convolutional code if
$$
d^c_j({\cal C})= (n-k)(j+1)+1,\;\; \mbox{for all } j\leq L.
$$
\end{Def}

%%%%%%%%%%%%%%%
Since the upper bounds of the column distances and the generalized Singleton bound coincide with the counterpart notions of $[\codc]_p$, if \codc\hs is a noncatastrophic convolutional code over $\ZZ[d]$ of rate $k/n $ and \bg\hs $\delta$ then \codc\hs also can achieve the upper bound for column distance only up to the instant $L$-\textit{th}. This leads to the following definition.
%%%%%%%%%%%%%%
\begin{Def}
Let $\codc$ be a noncatastrophic code over $\ZZ[d]$, of rate $k/n$ and $b$-degree $\delta$.
Let \mbox{$L=\left\lfloor \dfrac{\delta}{k}\right\rfloor + \left\lfloor \dfrac{\delta}{n-k}\right\rfloor$}.
\codc\hs is a Maximum Distance Profile (MDP) convolutional code if
$$
d^c_j({\cal C})= (n-k)(j+1)+1,\;\; \mbox{for all } j\leq L.
$$
\end{Def}

\begin{Thm}
Let $\codc$ be a noncatastrophic code over $\ZZ[d]$, of rate $k/n$ and $b$-degree $\delta$.
\codc\hs is an MDP code if and only if $[\codc]_p$ is an MDP code over $\Z_p[d]$.
\end{Thm}
\begin{proof}
Let \codc\hs be an MDP convolutional code, \mbox{$j\leq L=\left\lfloor \dfrac{\delta}{k}\right\rfloor + \left\lfloor \dfrac{\delta}{n-k}\right\rfloor$} and
$
\;\tilde{w}(d)=\displaystyle\sum_{i\in\mathbb{N}_0} \tilde{w}_id^i\in [\codc]_p,\;\tilde{w}_0 \neq 0$

Let us consider $w(d)$ as a vector of $\Z_{p^r}[n]^n$, and let $w(d)=\displaystyle\sum_{i\in\mathbb{N}_0} w_id^i; p^{r-1}\tilde{w}(d)\in\codc,\;w_0=p^{r-1}\tilde{w}_0 \neq 0$.
Since $wt([\;\tilde{w}_0\;\tilde{w}_1\;\ldots\;\tilde{w}_j\;])=wt([\;w_0\;w_i\;\ldots\;w_j\;])$ and \codc\hs is an MDP convolutional code and $d_j^c(\codc)\leq wt([\;w_0\;w_i\;\ldots\;w_j\;])$, it follows that
$$
(n-k)(j+1)+1\leq wt([\;\tilde{w}_0\;\tilde{w}_1\;\ldots\;\tilde{w}_j\;]),\; j\leq L
$$ and therefore $[\codc]_p$ is a MDP code.

Let us now consider $[\codc]_p$ a MDP code and $w(d)=\displaystyle\sum_{i\in\mathbb{N}_0} w_id^i\in \codc$, $w_0\neq 0$, $j\leq L=\left\lfloor \dfrac{\delta}{k}\right\rfloor + \left\lfloor \dfrac{\delta}{n-k}\right\rfloor$, $\ell$ the order of  $[\;w_0\;w_1\;\ldots\;w_j\;]$ and $s_1$ the smallest integer less than or equal to $j$, such that $w_{s_1}$ has order $\ell_1$. It follows that $p^{r-l}[\;w_0\;w_1\;\ldots\;w_j\;] = p^{r-1}[\;0\;0\;\tilde{w}_{s_1}\;\ldots\;\tilde{w}_j\;],\; \mbox{with}\; [\tilde{w}_{s_1}]_p \neq 0$.
Let $p^{r-l}w(d)=d^s( p^{r-1}\tilde{w}(d))$, \mbox{$\tilde{w}(d)=\displaystyle\sum_{i\in\mathbb{N}_0} \tilde{w}_{s_1+i}\,d^i$}, be a codeword. Since $G(d)$ is a \lzp\hs matrix, $G(0)$ is full row rank. So $d^s( p^{r-1}\tilde{w}(d))=(d^s u(d))G(d)$ and, therefore, $p^{r-1}\tilde{w}(d)= u(d)G(d)$, which implies that $p^{r-1}\tilde{w}(d)\in\codc$. That way, $[\tilde{w}(d)]_p\in [\codc]_p$ and
\begin{eqnarray*}
 wt\left([\;p^{r-1}\tilde{w}_{s_1}\;\ldots\;p^{r-1}\tilde{w}_{j-s_1}\;]\right) & = & wt\left([\;\tilde{w}_{s_1}\;\ldots\;\tilde{w}_{j-s_1}\;]_p\right),\quad [\tilde{w}_{s_1}]_p\neq 0.
\end{eqnarray*}
Furthermore, $[\codc]_p$ is MDP, so
$$
wt\left([\;p^{r-1}\tilde{w}_{s_1}\;\ldots\;p^{r-1}\tilde{w}_{j-s_1}+1\;]\right)\geq (n-k)(j-s_1+1)+1.
$$

Let us now consider $[\;w_0\;w_1\;\ldots\;w_{s_1-1}\;]$ with order $\ell_2\leq \ell$, such that
$$
p^{r-l_2}[\;w_0\;w_1\;\ldots\;w_{s_1-1}\;]=p^{r-1}[\;0\;0\;\tilde{\tilde{w}}_{s_2}\;\ldots\;\tilde{\tilde{w}}_{s_1-1}\;],
$$
with $s_2\leq s_1-1$ and $[\tilde{\tilde{w}}_{s_2}(d)]_p\neq 0$. Repeating the previous reasoning, we have to
$$
wt\left([\;p^{r-1}\tilde{w}_{s_2}\;\ldots\;p^{r-1}\tilde{w}_{s_1-1}\;]\right)\geq (n-k)(s_1-1+s_2+1)+1=(n-k)(s_1+s_2)+1
$$
Successively applying the previous process, we obtain
$$
wt\left([\;w_0\;w_1\;\ldots\;w_j\;]\right)\geq (n-k)(j+1)+1.
$$
\end{proof}

According to this theorem, we can easily get an MDP code over $\ZZ[d]$ of rate $k/n$ and $b$-degree $\delta$, from an MDP code over $\Z_p[d]$, of rate $k/n$ and degree $\delta$.

\section{Conclusions and future work}

In this paper we have investigated the central notion of primeness of polynomial matrices over $\ZZ$. We showed that zero left prime encoders define noncatastrophic convolutional codes over $\ZZ$ and allow a representation of the convolutional code by means of a polynomial parity-check matrix. We have studied free and columns distances of these codes and show that these are determined by the projection of the code over $\mathbb Z_p$. A natural and interesting avenue for future investigation is to generalized these results to wider classes of rings such as finite chain rings \cite{bouzara2020lifted}.

\section{Acknowledgments}
The second and third authors were supported by The Center for Research and Development in Mathematics and Applications (CIDMA) through the Portuguese
Foundation for Science and Technology (FCT - Fundação para a Ciência e a Tecnologia), references UIDB/04106/2020 and UIDP/04106/2020. Diego Napp is partially supported by Ministerio de Ciencia e Innovación via the grant with ref. PID2019-108668GB-I00.

%\end{paracol}
%\reftitle{References}

\bibliographystyle{plain}
\bibliography{biblio_com_tudo4}

%=====================================
% References, variant B: internal bibliography
%=====================================
\end{document}